\documentclass[1p,final]{elsarticle}
\usepackage{amssymb,pslatex}
\usepackage{amsthm,amsmath}
\usepackage{amssymb}
\usepackage{amsfonts}
\usepackage{latexsym}

\newtheorem{theorem}{Theorem}

\newtheorem{corollary}[theorem]{Corollary}
\newtheorem{example}[theorem]{Example}

\newtheorem{question}{Question}

\newcommand{\Tr}{{\mathrm{Tr}}}
\newcommand{\gf}{ {{\mathrm{GF}}} }

\newcommand{\PG}{ {{\mathrm{PG}}}}
\newcommand{\hyperoval}{ {{\mathcal{O}}}}

\newcommand{\Segre}{ {{\mathrm{Segre}}} }
\newcommand{\Glynnone}{ {{\mathrm{Glynni}}} }
\newcommand{\Glynntwo}{ {{\mathrm{Glynnii}}} }
\newcommand{\Payne}{ {{\mathrm{Payne}}} }
\newcommand{\Trans}{ {{\mathrm{Trans}}} }
\newcommand{\Adelaide}{ {{\mathrm{Adelaide}}} }
\newcommand{\Cherowitzo}{ {{\mathrm{Cherowitzo}}} }
\newcommand{\Subiaco}{ {{\mathrm{Subiaco}}} }

\begin{document}

\begin{frontmatter}



\title{Five Constructions of Permutation Polynomials over $\gf(q^2)$\tnotetext[fn1]{C. Ding's research was supported by The Hong Kong Research Grants Council, Proj. No. 601013. P. Yuan's research was supported by the NSF of China (Grant No.
11271142).}
}

\author[cding]{Cunsheng Ding}
\ead{cding@ust.hk}
\author[ypz]{Pingzhi Yuan}
\ead{mcsypz@mail.sysu.edu.cn}

\address[cding]{Department of Computer Science
                                                  and Engineering, The Hong Kong University of Science and Technology,
                                                  Clear Water Bay, Kowloon, Hong Kong, China}
\address[ypz]{School of Mathematics, South China Normal University, Guangzhou 510631, China}

\begin{abstract}
Four recursive constructions of permutation polynomials over $\gf(q^2)$ with those over
$\gf(q)$ are developed and applied to a few famous classes of permutation polynomials. 
They produce infinitely many new permutation polynomials over $\gf(q^{2^\ell})$ for any 
positive integer $\ell$ with any given permutation polynomial over $\gf(q)$.  
A generic construction of permutation polynomials over $\gf(2^{2m})$ with o-polynomials
over $\gf(2^m)$ is also presented, and a number of new classes of permutation polynomials 
over $\gf(2^{2m})$ are obtained. 
\end{abstract}


\end{frontmatter}

\section{Introduction}

It is well known that every function from $\gf(q)$ to $\gf(q)$ can be expressed as a polynomial over $\gf(q)$.
A polynomial $f \in \gf(q)[x]$ is called a \emph{permutation polynomial\index{permutation polynomial}} (PP) 
if the associated polynomial function $f: a \mapsto f(a)$ from $\gf(q)$ to $\gf(q)$ is a permutation of
$\gf(q)$.

Permutation polynomials have been a hot topic of study for many years, and have applications in coding theory 
\cite{Laigle,ST05}, cryptography \cite{Lidl85,LM84,Mullen89,RSA,SH98}, combinatorial designs \cite{DY06}, 
and other areas of mathematics and engineering. Permutation polynomials could have a huge impact in both 
theory and applications. For instance, the Dickson permutation polynomials of order five over $\gf(3^m)$, i.e., 
$D_5(x,a) = x^5 +ax^3-a^2x$, led to a 70-year research breakthrough in combinatorics \cite{DY06}, gave a 
family of perfect nonlinear functions for cryptography \cite{DY06}, generated good linear codes \cite{CDY,YCD06} 
for data communication and storage, and produced optimal signal sets for CDMA communications \cite{DY07}, 
to mention only a few applications of these permutation polynomials. A lot of progress on permutation polynomials 
has been made recently (see, for example, \cite{AGW,CHZ,Hou12,Houxd,HMSY,TZH,WY12,ZZH}, and the references 
therein). 

The objective of this paper is to present five constructions of permutation polynomials over $\gf(q^2)$ with 
those over $\gf(q)$. Four of them are recursive constructions, so that a sequence of permutation polynomials 
over $\gf(q^{2^\ell})$ can be constructed with only one permutation polynomial over $\gf(q)$, where $\ell$ 
is any positive integer. Another construction of permutation polynomials over $\gf(2^{2m})$ with o-polynomials 
over $\gf(2^m)$ is also developed. This nonrecursive construction gives several classes of new permutation 
polynomials over $\gf(2^{2m})$ with known classes of o-polynomials over $\gf(2^m)$.

\section{Some preparations}

Before presenting the five constructions of PPs over $\gf(q^2)$ from those over $\gf(q)$, 
we need to prove a few basic results as a preparation for subsequent sections.

\subsection{A general theorem about permutation polynomials}

The following is a fundamental result of this paper.

\begin{theorem}\label{main-thm2.1}
Let $\{\beta_1, \ldots, \beta_m\}$ and $\{\gamma_1, \ldots, \gamma_m\}$ be two bases of $\gf(q^m)$ over $\gf(q)$.
Let $A$ be an $m \times m$ nonsingular matrix over $\gf(q)$ and
$f_i(x)\in\gf(q)[x]$ for $i=1, \ldots, m$. Then
$$F(z)=\gamma_1f_1(x_1)+\ldots+ \gamma_mf_m(x_m),$$
where $(x_1, \ldots, x_m)=(z_1, \ldots, z_m)A$ and $z=\beta_1z_1+ \ldots +\beta_mz_m$, is a PP over $\gf(q^m)$ if and only if  all $f_i(x), i=1, \ldots, m$, are PPs over $\gf(q)$.
\end{theorem}

\begin{proof}
We first prove the sufficiency of the conditions and  assume now that all $f_i(x), i=1, \ldots, m$, are PPs over $\gf(q)$.
Let
$$z=\beta_1z_1+ \ldots +\beta_mz_m,  \ \ y=\beta_1y_1 + \ldots + \beta_my_m, \ \ z_i, y_i \in \gf(q),
$$
be two elements of $\gf(q^m)$ such that $F(z)=F(y)$. Let
$$(x_1, \ldots, x_m)=(z_1, \ldots, z_m)A$$
and
$$(u_1, \ldots, u_m)=(y_1, \ldots, y_m)A.$$
Then
$$\gamma_1f_1(x_1)+\ldots + \gamma_mf_m(x_m)=\gamma_1f_1(u_1)+\ldots + \gamma_mf_m(u_m).$$
Since $\{\gamma_1, \ldots, \gamma_m\}$ is a basis of $\gf(q^m)$ over $\gf(q)$, we deduce that $f_i(x_i)=f_i(u_i)$
for $i=1, \ldots, m$. Note that all $f_i(x), i=1, \ldots, m,$ are PPs over $\gf(q)$. We obtian that $x_i=u_i$,
and thus $z_i=y_i$, as $A$ is an $m \times m$ nonsingular matrix over $\gf(q)$. Therefore $z=y$, and hence $F(z)$ is a PP over $\gf(q^m)$.

We now prove the necessity of the conditions, and assume that $F(z)$ is a PP over $\gf(q^m)$.  Suppose on the contrary that some $f_i(x)$ is not PP over $\gf(q)$.
Without loss of generality, we may assume that $f_1(x)$ is not a PP over $\gf(q)$. Then there exist two distinct elements $x_1, u_1\in\gf(q)$ such that $x_1\ne u_1$ and $f_1(x_1)=f_1(u_1)$. Let
$$(y_1, \ldots, y_m)=(u_1, 0, \ldots, 0)A^{-1}, \quad (z_1, \ldots, z_m)=(x_1, 0, \ldots, 0)A^{-1}.$$
Then $(y_1, \ldots, y_m)\ne(z_1, \ldots, z_m)$ and $y \ne z$, since $A$ is an $m \times m$ nonsingular matrix over $\gf(q)$. Hence $F(z)=F(y)$ with $y \ne z$, which contradicts to our assumption that $F(z)$ is a PP over $\gf(q^m)$.
\end{proof}

\subsection{$\gf(q^2)$ as a two-dimensional space over $\gf(q)$}

In this subsection, we represent every element of $\gf(q^2)$ as a pair of elements in $\gf(q)$ 
with respect to a basis of  $\gf(q^2)$ over $\gf(q)$.

Let $r=q^2$, where $q$ is a power of a prime. Let $\beta$ be a generator of $\gf(r)$. Then $\{1, \beta\}$ is a basis of
$\gf(r)$ over $\gf(q)$, and any element $z \in \gf(r)$ can be expressed as
$$
z=x+y\beta,
$$
where
\begin{eqnarray}\label{eqn-basisexpression}
x=\frac{\beta^q z - \beta z^q}{\beta^q-\beta} \in \gf(q), \ \
y=\frac{z^q-z}{\beta^q-\beta} \in \gf(q).
\end{eqnarray}

\section{A construction of permutation polynomials over $\gf(q^2)$ with o-polynomials over $\gf(q)$}

Throughout this section, let $q=2^m$, where $m>1$ is a positive integer.
A permutation polynomial $f$ on $\gf(q)$ is called an \emph{o-polynomial\index{o-polynomial}} if
$f(0)=0$, and for each $s \in \gf(q)$,
\begin{eqnarray}
f_s(x)=(f(x+s)+f(s))x^{q-2}
\end{eqnarray}
is a permutation polynomial. In the original definition of o-polynomials, it is required that $f(1)=1$.
However, this is not essential, as one can always normalise $f(x)$ by using $f(1)^{-1} f(x)$ due to
that $f(1) \neq 0$.

A \emph{hyperoval} of the projective plane $\PG(2, q )$ is a set of $q + 2$ points such
that no three of them are collinear. Any hyperoval in $\PG(2, q)$ can be 
written as 
\begin{eqnarray}\label{eqn-hyperoval}
\hyperoval_f=\{(1, t, f(t)): t \in \gf(q)\} \cup \{(0,1,0), (0,0,1)\}
\end{eqnarray}
where $f$ is a o-polynomial over $\gf(q)$ with $f(0)=0$ and $f(1)=1$.  Two o-polynomials are called 
{\em equivalent} if their hyperovals are equivalent. 

In this section, we will present a construction of PPs over $\gf(q^2)$ with o-polynomials over $\gf(q)$
and will give a summary of known  o-polynomials over $\gf(q)$.

\subsection{The construction}\label{sec-constwithopoly}

Let $f$ be a polynomial over $\gf(q)$ and let $\beta$ be a generator of $\gf(q^2)^*$. We now define
a polynomial $F(z)$ over $\gf(q^2)$ by
\begin{eqnarray}\label{eqn-bignewF}
F(z) &=& xf(yx^{q-2}) + \beta y + \left((z+z^q)^{q-1}+1\right)z \nonumber \\
       &=& \frac{\beta^q z+\beta z^q}{\beta^q+\beta} f \left[  \left(  \frac{\beta^q z+\beta z^q}{\beta^q+\beta} \right)
              \left( \frac{z+z^q}{\beta^q+\beta} \right)^{q-2}  \right]  + \left((z+z^q)^{q-1}+1\right)z,
\end{eqnarray}
where $z=x+\beta y$,  $x \in \gf(q)$, $y \in \gf(q)$ and $z$ are also related by the expression of (\ref{eqn-basisexpression}).

\begin{theorem}
If $f$ is an o-polynomial over $\gf(q)$, then the polynomial $F(x)$ of (\ref{eqn-bignewF}) is a PP over $\gf(q^2)$.
\end{theorem}

\begin{proof}
Since $z+z^q=(\beta+\beta^q)y\in\gf(q)$ and $z+z^q=0$ if and only if $y=0$, so $(z+z^q)^{q-1}=1$
when $y\ne0$ and $0$ otherwise. Hence $F(z)=xf(yx^{q-2}) + \beta y$ when $y\ne0$, and $F(z)=x$ when $y=0$.

Let $z_1=x_1+\beta y_1$ and $z_2=x_2+\beta y_2$. Assume that $F(z_1)=F(z_2)$. We first consider the case
that $y_1=0$ and $y_2=0$. In this case, we have
$$F(z_1)=x_1=x_2=F(z_2).$$
Hence $x_1=x_2$, and then $z_1=z_2$. Therefore we may assume that $y_1\ne 0$. In this case, we have
 $$F(z_1)=x_1f(y_1x_1^{q-2}) + \beta y_1= F(z_2)=x_2f(y_2x_2^{q-2}) + \beta y_2.$$
Hence $y_2\ne0$ and $y_1=y_2=y$ since $\{1, \beta\}$ is a basis of $\gf(q^2)$ over $\gf(q)$. It follows that
\begin{equation}\label{o1}
x_1f(yx_1^{q-2})=x_2f(yx_2^{q-2}).
\end{equation}
Since $f$ is an o-polynomial over $\gf(q)$, by definition,  $f_0(x):=f(x)x^{q-2}$ is a permutation polynomial over $\gf(q)$. Note that $yx^{q-2}$ is also a  permutation polynomial over $\gf(q)$, which implies that $f_0(yx^{q-2})=y^{q-2}xf(yx^{q-2})$ is a  permutation polynomial over $\gf(q)$. Hence $xf(yx^{q-2})$ is a  permutation polynomial over $\gf(q)$. It follows from (\ref{o1}) that $x_1=x_2$,  and so $z_1=z_2$. This proves the theorem.
\end{proof}

In the construction of this section, the o-polynomial property of $f(x)$ is sufficient. But it is open if it is necessary.

\subsection{Known o-polynomials over $\gf(2^m)$}

To obtain permutation polynomials over $\gf(2^{2m})$ with the generic construction of Section \ref{sec-constwithopoly}, 
we need explicit o-polynomials over $\gf(2^m)$ as building blocks. O-polynomials have many special properties. For instance, 
the coefficient of each term of odd power in an o-polynomial is zero \cite{Okeefe}. 

In this subsection, we summarize and extend known 
o-polynomials over $\gf(2^m)$. We also introduce some basic results about o-polynomials.

\subsubsection{Basic properties of o-polynomials}

For any permutation polynomial $f(x)$ over $\gf(2^m)$, we define $\overline{f}(x)=xf(x^{2^m-2})$, and use $f^{-1}$ to denote the compositional inverse of $f$, i.e., $f^{-1}(f(x))=x$ for all $x \in \gf(2^m)$.

The following theorem introduces basic properties of o-polynomials whose proofs are easy \cite{Okeefe}. 

\begin{theorem}\label{thm-basicproperty}
Let $f$ be an o-polynomial on $\gf(2^m)$. Then the following holds:
\begin{enumerate}[a)]
\item $f^{-1}$ is also an o-polynomial;
\item $f(x^{2^{j}})^{2^{m-j}}$ is also an o-polynomial for any $1 \leq j \leq m-1$;
\item $\overline{f}$ is also an o-polynomial; and
\item $f(x+1)+f(1)$ is also an o-polynomial.
\end{enumerate}
\end{theorem}

Although these o-polynomials are equivalent to the original o-polynomial $f$ in terms of their hyperovals, 
they may produce different objects when they are used in other applications. 

All o-monomials are characterised by the following theorem \cite[Corollary 8.2.4]{Hirsh}.  

\begin{theorem}
The monomial $x^k$ is an o-polynomial over $\gf(2^m)$ if and only if 
\begin{enumerate}[a)]
\item $\gcd(k, 2^m-1)=1$; 
\item $\gcd(k-1, 2^m-2)=1$; and 
\item $((x+1)^k+1)x^{2^m-2}$ is a permutation polynomial over $\gf(2^m)$. 
\end{enumerate}
\end{theorem}

For o-monomials we have the following fundamental result \cite{Hirsh}. 

\begin{theorem}\label{thm-basicproperty2}
Let $x^k$ be an o-polynomial on $\gf(2^m)$. Then every polynomial in
$$\left\{x^{\frac{1}{k}},\, x^{1-k},\,  x^{\frac{1}{1-k}},\, x^{\frac{k}{k-1}},\, x^{\frac{k-1}{k}}\right\}$$
is also an o-monomial, where $1/k$ denotes the multiplicative inverse of $k$ modulo $2^m-1$.
\end{theorem}

\subsubsection{The translation o-polynomials}

The translation o-polynomials are described in the following theorem \cite{Segre57}.

\begin{theorem}
$\Trans(x)=x^{2^h}$ is an o-polynomial on $\gf(2^m)$, where $\gcd(h, m)=1$.
\end{theorem}

The following is a list of known properties of translation o-polynomials.

\begin{enumerate}[1)]
\item $\Trans^{-1}(x)=x^{2^{m-h}}$ and
\item $\overline{\Trans}(x)= xf(x^{2^m-2})=x^{2^m-2^{m-h}}$.
\end{enumerate}

\subsubsection{The Segre and Glynn o-polynomials}

The following theorem describes a class of o-polynomials, which are an extension of the original Segre
o-polynomials.

\begin{theorem}
Let $m$ be odd. Then
$\Segre_a(x)=x^{6}+ax^4+a^2x^2$ is an o-polynomial on $\gf(2^m)$ for every $a \in \gf(2^m)$.
\end{theorem}

\begin{proof}
The conclusion follows from
$$
\Segre_a(x)=(x+\sqrt{a})^6 + \sqrt{a}^3.
$$
\end{proof}

We have the following remarks on this family of o-polynomials.
\begin{enumerate}[a)]
\item $\Segre_0(x)=x^6$ is the original Segre o-polynomial \cite{Segre62,SegreBartocci}. So this is an extended family.
\item $\Segre_a(x)=xD_5(x, a)=a^2D_5(x^{2^m-2}, a^{2^m-2})x^7$, where $D_5(x, a)=x^5+ax^3+a^2x$,
          which is the Dickson polynomial of the first kind of order 5 \cite{LMT93}.
\item $\overline{\Segre}_a=D_5(x^{2^m-2}, a)=a^2x^{2^m-2}+ax^{2^m-4}+x^{2^m-6}$.
\item $\Segre_a^{-1}(x)=(x+\sqrt{a}^3)^{\frac{5\times 2^{m-1}-2}{3}} +\sqrt{a}$.
\end{enumerate}

The proof of the following theorem is straightforward and omitted. 

\begin{theorem}\label{thm-jan9}
Let $m$ be odd. Then
\begin{eqnarray}\label{eqn-PayneInverse2}
\overline{\Segre}_1^{-1}(x)=\left( D_{\frac{3 \times 2^{2m}-2}{5}}(x, 1)\right)^{2^m-2}.
\end{eqnarray}
\end{theorem}

Glynn discovered two families of o-polynomials \cite{Glynn83}. The first is described as follows.

\begin{theorem}
Let $m$ be odd. Then $\Glynnone(x)=x^{3\times 2^{(m+1)/2}+4}$ is an o-polynomial.
\end{theorem}

An extension of the second family of o-polynomials discovered by Glynn is documented in the following theorem.

\begin{theorem}
Let $m$ be odd. Then
\begin{eqnarray*}
\Glynntwo_a(x)= \left\{
\begin{array}{ll}
x^{2^{(m+1)/2}+2^{(3m+1)/4}} + a x^{2^{(m+1)/2}} + (ax)^{2^{(3m+1)/4}}   & \mbox{ if } m \equiv 1 \pmod{4}, \\
x^{2^{(m+1)/2}+2^{(m+1)/4}} + a x^{2^{(m+1)/2}} + (ax)^{2^{(m+1)/4}}           & \mbox{ if } m \equiv 3 \pmod{4}.
\end{array}
\right.
\end{eqnarray*}
 is an o-polynomial for all $a \in \gf(q)$.
\end{theorem}

\begin{proof}
Let $m \equiv 1 \pmod{4}$. Then
$$
\Glynntwo_a(x)=(x+a^{(m-1)/4})^{2^{(m+1)/2}+2^{(3m+1)/4}} + a^{2^{(m+1)/2}+2^{(3m+1)/4}}.
$$
The desired conclusion for the case $m \equiv 1 \pmod{4}$ can be similarly proved.
\end{proof}

Note that $\Glynntwo_0(x)$ is the original Glynn o-polynomial. So this is an extended family. For some applications,
the extended family may be useful.

\subsubsection{The Cherowitzo o-polynomials}

The following describes another class of o-polynomials.

\begin{theorem}
Let $m$ be odd and $e=(m+1)/2$. Then
$$
\Cherowitzo_a(x)=x^{2^e}+ax^{2^e+2}+a^{2^e +2}x^{3 \times 2^e +4}
$$
is an o-polynomial on $\gf(2^m)$ for every $a \in \gf(2^m)$.
\end{theorem} 

We have the following remarks on this family.
\begin{enumerate}[1)]
\item $\Cherowitzo_1(x)$ is the original Cherowitzo o-polynomial \cite{Ch88,Ch96}. So this is an extended family.
\item A proof of the o-polynomial property of this extended family goes as follows. It can be easily verified that 
$\Cherowitzo_a(x)=a^{-2^{e-1}} \Cherowitzo_1(a^{1/2}x)$. The desired conclusion then follows.  
\item $\overline{\Cherowitzo}(x)=x^{2^m-2^e}+ax^{2^m-2^e-2}+a^{2^e +2}x^{2^m-3 \times 2^e -4}$.
\item It is known that $\Cherowitzo_1^{-1}(x)=x(x^{2^e+1}+x^3+x)^{2^{e-1}-1}$.
\end{enumerate}

The proofs of the following two theorems are straightforward and left to the reader.

\begin{theorem}
$$
\Cherowitzo_a^{-1}(x)=x(ax^{2^e+1}+a^{2^e}x^3+x)^{2^{e-1}-1}.
$$
\end{theorem}

\begin{theorem}
$$
\overline{\Cherowitzo}=(ax^{2^m-2^e-2}+a^{2^e}x^{2^m-4}+x^{2^m-2})^{2^{e-1}-1}.
$$
\end{theorem}

\subsubsection{The Payne o-polynomials}

The following documents an extended family of o-trinomials.

\begin{theorem}\label{thm-Payne11}
Let $m$ be odd. Then
$\Payne_a(x)=x^{\frac{5}{6}}+ax^{\frac{3}{6}}+a^2x^{\frac{1}{6}}$ is an o-polynomial on $\gf(2^m)$ for every $a \in \gf(2^m)$, where $\frac{1}{6}$ denotes the multiplicative inverse of $6$ modulo $2^m-1$. 
\end{theorem}

We have the following remarks on this family.
\begin{enumerate}[a)]
\item $\Payne_1(x)$ is the original Payne o-polynomial \cite{Pay85}. So this is an extended family. 
\item It can be verified that $\Payne_a(x)=a^{5/2}\Payne_1(a^{-3}x)$. The desired conclusion of 
of Theorem \ref{thm-Payne11} follows.  
\item $\Payne_a(x)=xD_5(x^{\frac{1}{6}}, a)$.
\item $\overline{\Payne}_a(x)=a^{2^m-3}\Payne_{a^{2^m-2}}(x)$.
\item Note that
$$
\frac{1}{6}=\frac{5 \times 2^{m-1}-2}{3}.
$$
We have then
$$
\Payne_a(x)=x^{\frac{2^{m-1}+2}{3}} + ax^{2^{m-1}} + a^2x^{\frac{5 \times 2^{m-1}-2}{3}}.
$$
\end{enumerate}

\begin{theorem}
Let $m$ be odd. Then
\begin{eqnarray}\label{eqn-PayneInverse}
\Payne_1^{-1}(x)=\left( D_{\frac{3 \times 2^{2m}-2}{5}}(x, 1)\right)^6
\end{eqnarray}
and $\overline{\Payne}_1^{-1}(x)$ are an o-polynomial.
\end{theorem}

\begin{proof}
Note that the multiplicative inverse of $5$ modulo $2^m-1$ is $\frac{3 \times 2^{2m}-2}{5}$. The conclusion then
follows from the definition of the Payne polynomial and the fact that
$$
D_5(x, 1)^{-1}=D_{\frac{3 \times 2^{2m}-2}{5}}(x, 1).
$$
\end{proof}

\subsubsection{The Subiaco o-polynomials}

The Subiaco o-polynomials are given in the following theorem \cite{Subiaco}.

\begin{theorem}\label{thm-Subiaco}
Define
$$
\Subiaco_a(x)=((a^2(x^4+x)+a^2(1+a+a^2)(x^3+x^2)) (x^4 + a^2 x^2+1)^{2^m-2}+x^{2^{m-1}},
$$
where $\Tr(1/a)=1$ and $d \not\in \gf(4)$ if $m \equiv 2 \bmod{4}$. Then $\Subiaco_a(x)$ is an
o-polynomial on $\gf(2^m)$.
\end{theorem}

As a corollary of  Theorem \ref{thm-Subiaco}, we have the following.

\begin{corollary}
Let $m$ be odd. Then
\begin{eqnarray}\label{cor-Subiaco}
\Subiaco_1(x)=(x+x^2+x^3+x^4) (x^4 + x^2+1)^{2^m-2}+x^{2^{m-1}}
\end{eqnarray}
is an o-polynomial over $\gf(2^m)$.
\end{corollary}

\subsubsection{The Adelaide o-polynomials}

The last known family of o-polynomials were discovered in \cite{Adeil} and described in the following theorem.

\begin{theorem}
Let $m$ be even and $r=q^2$. The trace function from $\gf(r)$ to $\gf(q)$ is
defined by $\Tr_{2}(x)=x+x^q$. Let $b \in \gf(r)$ such that $b^{q+1}=1$ and $b \neq 1$ and define
\begin{eqnarray*}
\lefteqn{\Adelaide_b(x) =} \\
&  \Tr_{2}(b)^{q-2}  \Tr_{2}(b^\ell) (x+1) +
                       \Tr_{2}(b)^{q-2}  \Tr_{2}((bx+b^q)^{r_1+\ell}) (x+\Tr_2(b)x^{2^{m-1}}+1)^{q-\ell} +x^{2^{m-1}},
\end{eqnarray*}
where $\ell = \pm (q-1)/3$. Then $f_b$ is an o-polynomial.
\end{theorem}

We need to clarify the definition of this family of o-polynomials. First of all, $\Adelaide_b(x)$ is defined as a
polynomial over $\gf(q)$, while $b$ is an element from the extension field. Secondly, note that $\ell$ could be negative. The
definition above is a modified version of the original one in the literature.

\section{The first recursive construction}\label{sec-1stc}

We now present the first recursive construction of PPs over $\gf(q^2)$ from those over $\gf(q)$.
Let $f_1(x)$ and $f_2(x)$ be two polynomials over $\gf(q)$, and let $\beta \in \gf(q^2) \setminus \gf(q)$. 
Note that $\{1, \beta\}$ is a basis of
$\gf(q^2)$ over $\gf(q)$. Let $z=x+\beta y$, where  $x \in \gf(q)$, $y \in \gf(q)$
and $z$ are also related by the expression of (\ref{eqn-basisexpression}). It follows from Theorem
\ref{main-thm2.1} that
\begin{eqnarray}\label{eqn-bignewF11}
F_1(z) := f_1(x) + \beta f_2(y)
= f_1 \left( \frac{\beta^q z - \beta z^q}{\beta^q-\beta} \right) + \beta \  f_2 \left( \frac{z^q-z}{\beta^q-\beta} \right).
\end{eqnarray}
is a permutation polynomial over $\gf(r)$ if and only if both $f_1(x)$ and $f_2(x)$ are  permutation polynomials over $\gf(q)$. This proves the following theorem.

\begin{theorem}\label{thm-ourpp1}
Let $f_1(x)$ and $f_2(x)$ be  two polynomials over $\gf(q)$. Then the polynomial $F_1(z)$ of (\ref{eqn-bignewF11})
is a permutation polynomial over $\gf(q^2)$ if and only if both $f_1(x)$ and $f_2(x)$ are  permutation polynomials over
$\gf(q)$.
\end{theorem}

As a special case of Theorem \ref{thm-ourpp1}, let $f_1(x)=f_2(x)=f(x)$, where $f(x)$ is a PP over $\gf(q)$. Then this
$g(x)$ gives a PP over $\gf(q^2)$, and the newly obtained PP over $\gf(q^2)$ gives another new PP over $\gf(q^4)$.
By recursively applying this construction, we obtain a PP over $\gf\left(q^{2^i}\right)$ for any integer $i$.

As a demonstration of the generic construction of this section, we consider a few special cases below. We start with the
so-called $p$-polynomials. Let $q=p^m$ for some $m$. A $p$-polynomial $L(x)$ over $\gf(q)$ is of the form
$$
L(x)=\sum_{i=0}^{\ell-1} l_i x^{p^i},
$$
where the coefficients $l_i \in \gf(q)$. It is known that $L(x)$ is a PP over $\gf(q)$ if and only if $L(x)$ has only the
root $0$ in $\gf(q)$ \cite[p. 351]{LN97}. However, this characterization is not really useful for constructing permutation
$q$-polynomials. The following corollary of Theorem \ref{thm-ourpp1} shows that the generic construction of this section
can be employed to construct permutation $p$-polynomials.

\begin{corollary}\label{cor-ourcorr1}
Let $q=p^m$ for some positive integer $m$. Let $\beta \in \gf(q^2) \setminus \gf(q)$. If $f_1(x)$
and $f_2(x)$ are two permutation $p$-polynomials over $\gf(q)$, then the polynomial $F_1(z)$ of 
(\ref{eqn-bignewF11}) is a permutation $p$-polynomial over $\gf(q^2)$.
\end{corollary}

As a special case of Corollary \ref{cor-ourcorr1}, let
$$
f_1(x)=ax^{p^{h_1}} \mbox{ and } f_2(x)=ax^{p^{h_2}},
$$
where $a, \, b \in \gf(q)^*$ and $0 \leq h_i \leq m-1$ for each $i$. Then both $f_1(x)$ and $f_2(x)$ are permutation
$p$-polynomials over $\gf(q)$. By Corollary \ref{cor-ourcorr1},
$$
a\left( \frac{\beta^q z - \beta z^q}{\beta^q-\beta} \right)^{p^{h_1}} +
\beta  b \left( \frac{z^q-z}{\beta^q-\beta} \right)^{p^{h_2}}
$$
is a permutation $p$-polynomial over $\gf(q^2)$.

Note that the monomial $ax^u$ is a PP over $\gf(q)$ if and only if $a \neq 0$ and $\gcd(u, q-1)=1$.
 If we choose $f_1(x)$ and $f_2(x)$ in Theorem \ref{thm-ourpp1} as monomials, we obtain the following.

\begin{corollary}\label{cor-ourpp222}
Let $1 \leq u \leq q-1$ and $1 \leq v \leq q-1$ be two integers, and let $\beta \in \gf(q^2) \setminus \gf(q)$. 
Then 
$$\eta \left(\frac{\beta^q z - \beta z^q}{\beta^q-\beta}\right)^u+ \gamma \beta\left(\frac{z^q-z}{\beta^q-\beta}\right)^v$$
is a PP over $\gf(q^2)$ if and only if $\gcd(uv, q-1)=1$, where $\eta \in \gf(q)^*$ and  $\gamma \in \gf(q)^*$.
\end{corollary}

Any PP $f(x)$ over $\gf(q)$ can be plugged into the generic construction of this section to obtain a PP
over $\gf(q^2)$. So it is endless to consider all the specific constructions in this paper. However, it would
be interesting to investigate the specific permutation polynomials over $\gf(q^2)$ from the Dickson
permutation polynomials of the first kind over $\gf(q)$, which are defined by
\begin{eqnarray}\label{eqn-1stDP}
D_h(x, a)=\sum_{i=0}^{\lfloor \frac{h}{2} \rfloor} \frac{h}{h-i} \binom{h-i}{i} (-a)^i x^{h-2i},
\end{eqnarray}
where $a \in \gf(q)$ and $h$ is called the {\em order} of the polynomial. It is known that $D_h(x, a)$ is
a PP over $\gf(q)$ if and only if $\gcd(h, q^2-1)=1$. 

As a corollary of Theorem \ref{thm-ourpp1}, we have the following. 

\begin{corollary} 
Let $\beta \in \gf(q^2) \setminus \gf(q)$ and $a \in \gf(q)$. 
Then the following polynomial 
\begin{eqnarray}\label{eqn-1stDPgfq2}
\overline{D}_h(x, a)=\sum_{i=0}^{\lfloor \frac{h}{2} \rfloor} \frac{h}{h-i} \binom{h-i}{i} (-a)^i 
\left[    \left(\frac{\beta^q z - \beta z^q}{\beta^q-\beta}\right)^{h-2i}  + 
             \beta \left(\frac{z^q-z}{\beta^q-\beta}\right)^{h-2i} \right],
\end{eqnarray}
is
a PP over $\gf(q^2)$ if and only if $\gcd(h, q^2-1)=1$. 
\end{corollary}

\section{The second recursive construction}\label{sec-2ndc}

The second recursive construction is a variant of Theorem \ref{thm-ourpp1} and is described below.

\begin{theorem}
Let $\beta \in \gf(q^2) \setminus \gf(q)$. 
Let $f_1(x)$ and $f_2(x)$ be  two polynomials over $\gf(q)$. Then
$$
F_2(z):=f_1\left( \frac{\beta^q z - \beta z^q}{\beta^q-\beta} \right) +
\beta f_2 \left( \frac{(\beta^q-1)z - (\beta-1)z^q}{\beta^q-\beta} \right)
$$
is a PP over $\gf(q^2)$ if and only if both $f_1(x)$ and $f_2(x)$ are PPs over $\gf(q)$.
\end{theorem}

\begin{proof}
Let $x_1=x$ and $x_2=x+y$. The proof is similar to that of Theorem \ref{thm-ourpp1} and is omitted.
\end{proof}

This is not only a recursive but also a generic construction, into which any permutation polynomials 
$f_1(x)$ and $f_2(x)$ over 
$\gf(q)$ can be plugged.

\section{The third recursive construction}

Let $\beta \in \gf(q^2) \setminus \gf(q)$. 
Note that $\{\beta+1, \beta\}$ is also a basis over  $\gf(q)$.  Let $f_1(x)$ and $f_2(x)$ be two polynomials over $\gf(q)$. We now define a polynomial $F_3(z)$ over $\gf(q^2)$ by
\begin{eqnarray}\label{eqn-bignewF3}
F_3(z) &=& f_1(x) + \beta (f_1(x)+f_2(y))  \nonumber \\
       &=& (\beta+1)f_1\left( \frac{\beta^q z - \beta z^q}{\beta^q-\beta} \right) + \beta f_2 \left( \frac{z^q-z}{\beta^q-\beta} \right),
\end{eqnarray}
where $z=x+\beta y$,  $x \in \gf(q)$, $y \in \gf(q)$ and $z$ are also related by the expression of (\ref{eqn-basisexpression}).

The following theorem then follows from Theorem \ref{main-thm2.1}.

\begin{theorem}
Let $\beta \in \gf(q^2) \setminus \gf(q)$, and let
$f_1(x)$ and $f_2(x)$ be  two polynomials over $\gf(q)$. Then the polynomial $F_3(z)$ of (\ref{eqn-bignewF3})
is a PP over $\gf(q^2)$ if and only if both $f_1(x)$ and $f_2(x)$ are PPs over $\gf(q)$.
\end{theorem}

\begin{example}
Let $q$ be an odd prime power.  Choose an element $b\in\gf(q)$ such that $x^2-b$ is irreducible
over $\gf(q)$. Let $\beta$ be a solution of $x^2-b$ in $\gf(q^2)$ (we view $\gf(q^2)$ as the splitting field of
$x^2-b$ over $\gf(q)$). Let
$$
x_1=a_{11}x+2ba_{12}y \mbox{ and } x_2= a_{21}x+2ba_{22}y,
$$
where $a_{11},\, a_{12},\, a_{21},\, a_{22} \in \gf(q)$ with $a_{11}a_{22}-a_{12}a_{21}\ne 0$.
Note that $\{1, \beta\}$, $\beta \in \gf(r) \setminus \gf(q)$ is a basis of $\gf(q^2)$ over $\gf(q)$.
It then follows from Theorem \ref{main-thm2.1} that
$$\left((a_{11}+a_{12}\beta)z+(a_{11}-a_{12}\beta)z^q\right)^u+\alpha\left((a_{21}+a_{22}\beta)z+(a_{21}-a_{22}\beta)z^q\right)^v, \, \alpha \in \gf(r)\setminus \gf(q)$$
is a PP over $\gf(q^2)$ if and only if $\gcd(uv, q-1)=1$.
\end{example}

\section{The fourth recursive construction}

Throughout this section, let $q$ be a power of $2$ and let $\beta \in \gf(q^2) \setminus \gf(q)$.
Let $f$ be a polynomial over $\gf(q)$. We now define a polynomial $G(z)$ over $\gf(q^2)$ by
\begin{eqnarray}\label{eqn-bignewG}
G(z) &=& f(yx) + \beta y + \left((z+z^q)^{q-1}+1\right)z \nonumber \\
       &=&  f \left[  \left(  \frac{\beta^q z+\beta z^q}{\beta^q+\beta} \right)
              \left( \frac{z+z^q}{\beta^q+\beta} \right)  \right]  + \left((z+z^q)^{q-1}+1\right)z,
\end{eqnarray}
where $z=x+\beta y$,  $x \in \gf(q)$, $y \in \gf(q)$ and $z$ are also related by the expression of (\ref{eqn-basisexpression}).

\begin{theorem}
If $f$ is a permutation polynomial over $\gf(q)$, then the polynomial $G(z)$ of (\ref{eqn-bignewG}) is a PP
over $\gf(q^2)$.
\end{theorem}

\begin{proof}
Since $q$ is a power of 2, we have that $z+z^q=(\beta+\beta^q)y\in\gf(q)$.
Thus, $z+z^q=0$ if and only if $y=0$. It then follows that $(z+z^q)^{q-1}=1$ when $y\ne0$
and $0$ otherwise. Consequently, $F(z)=f(yx) + \beta y$ when $y \ne 0$ and $F(z)=x$ when $y=0$.

Let $z_1=x_1+\beta y_1$ and  $z_2=x_2+\beta y_2$. Assume that $F(z_1)=F(z_2)$. We first consider the case
that $y_1=0$ and $y_2=0$.  In this case, we have
$$F(z_1)=x_1=x_2=F(z_2).$$
Thus, $x_1=x_2$, and then $z_1=z_2$. Therefore we may assume
that $y_1 \ne 0$. We have
$$F(z_1)=f(y_1x_1) + \beta y_1= F(z_2)=f(y_2x_2) + \beta y_2.$$
As a result, we have $y_2 \ne 0$ and $y_1=y_2=y$ since $\{1, \beta\}$ is a basis of $\gf(q^2)$ over $\gf(q)$. It follows that
\begin{equation}\label{o2}
f(yx_1)=f(yx_2).
\end{equation}
Note that $y \ne 0$.
Since $f(x)$ is a permutation polynomial over $\gf(q)$, so is $f(yx)$. It follows from (\ref{o2}) that $x_1=x_2$.  Hence,  $z_1=z_2$. This proves the theorem.
\end{proof}

Notice that the construction of this section works only for the case that $q$ is a power of 2.

\section{Further constructions of permutation polynomials}

The generic  constructions of permutation polynomials presented in some earlier sections are derived 
from Theorem \ref{main-thm2.1}. In this section, we employ Theorem \ref{main-thm2.1} to get 
more constructions of permutation polynomials over $\gf(q^2)$ and $\gf(q^3)$.

\begin{theorem}\label{thm-today151} 
Let $q$ be an odd prime power, $t$ a positive integer, and let $a, b, u \in \gf(q)$. Then
$$F(z)=az+bz^q+(z+z^q+u)^t$$
is a PP over $\gf(q^2)$ if and only if $a\ne b$ and $(a+b)x+2x^t$ is a PP over $\gf(q)$.
\end{theorem}

\begin{proof} 
Since $q$ is an odd prime power, $2(q-1)$ must divide $q^2-1$. Consequently, there exists an element $\alpha\in\gf(q^2)\backslash\gf(q)$ with $\alpha^q=-\alpha$. Thus $\{1, \alpha\}$ is a basis of $\gf(q^2)$ over $\gf(q)$. Put
$$z=x+y\alpha,$$
then $z^q=x-y\alpha$. It follows that
$$F(z)=(a+b)x+(a-b)y\alpha+(2x+u)^t.$$
By Theorem \ref{main-thm2.1}, $F(z)$ is a PP over $\gf(q^2)$ if and only if both $(a-b)x$ and $(a+b)x+(2x+u)^t$ are PPs over $\gf(q)$. It is easily seen that $(a+b)x+(2x+u)^t$ is a PP over $\gf(q)$ if and only if $(a+b)x+2x^t$ is a PP over $\gf(q)$.
The desired conclusions then follow.  
\end{proof}

When $t=1$, $a \ne b$ and $a+b+2 \ne 0$, it follows from Theorem \ref{thm-today151} that 
$F(z)=(a+1)z+(b+1)z^q+u$ is a PP over $\gf(q^2)$ for any $u \in \gf(q)$.

\begin{theorem} 
Let $q$ be an odd prime power and let $t$ a positive integer. Let $\alpha \in \gf(q^2)\backslash\gf(q)$ with 
$\alpha^q=-\alpha$ and let $a, b, u \in \gf(q)$.

\begin{enumerate}
\item When $t$ is even,  
$$F_1(z)=az+bz^q+(z-z^q+u\alpha)^t$$
is a PP over $\gf(q^2)$ if and only if $a^2\ne b^2$. 

\item When $t$ is odd,  
$$F_1(z)=az+bz^q+(z-z^q+u\alpha)^t$$
is a PP over $\gf(q^2)$ if and only if $a+b\ne0$ and $(a-b)x+2x^t\alpha^{t-1}$ is a PP over $\gf(q)$.
\end{enumerate}
\end{theorem}

\begin{proof} 
We first prove the conclusion of the first part. Since $\alpha^q=-\alpha$, $\{1, \alpha\}$ is a basis of $\gf(q^2)$ over $\gf(q)$ and $\alpha^2\in\gf(q)$. 
Put
$$z=x+y\alpha,$$
where $x \in \gf(q)$ and $y \in \gf(q)$. Then $z^q=x-y\alpha$. It follows that
\begin{equation}\label{eq91} 
F_1(z)=(a+b)x+(a-b)y\alpha+(2y+u)^t\alpha^t. 
\end{equation}

If $a+b=0$, then $F_1(x_1+y\alpha)=F_1(x_2+y\alpha)$ for any $x_1, x_2, y\in\gf(q)$. Thus $F_1(z)$ is not a PP over $\gf(q^2)$. 

Now we assume that $a+b \ne 0$. Let 
$$ 
z_1=x_1+y_1\alpha, \, \, z_2=x_2+y_2\alpha, \, \, x_1, \, x_2,\,  y_1,\, y_2\in\gf(q)
$$ 
such that $F_1(z_1)=F_1(z_2)$. 

By (\ref{eq91}) and $\alpha^t\in\gf(q)$, we have
$$(a-b)y_1=(a-b)y_2, \quad (a+b)x_1+(2y_1+u)^t\alpha^t=(a+b)x_2+(2y_2+u)^t\alpha^t.$$
If $a=b$, then for any distinct $y_1, y_2\in\gf(q)$, we can obtain two elements $x_1, x_2\in\gf(q)$ with 
$$(a+b)x_1+(2y_1+u)^t\alpha^t=(a+b)x_2+(2y_2+u)^t\alpha^t.$$Thus $F_1(z)$ is not a PP over $\gf(q^2)$.

If $a \ne b$, then we have $y_1=y_2$ and $(a+b)x_1=(a+b)x_2$, so $x_1=x_2$. It follows that $z_1=z_2$, which implies that $F_1(z)$ is a PP over $\gf(q^2)$ if and only if $a^2\ne b^2$.

Now we turn to prove the conclusion of the second part.. By (\ref{eq91}) and $\alpha^{t-1}\in\gf(q)$, we have
$$(a+b)x_1=(a+b)x_2, \quad (a-b)y_1+(2y_1+u)^t\alpha^{t-1}=(a-b)y_2+(2y_2+u)^t\alpha^{t-1}.$$
Hence by Theorem 2.1, $F(z)$ is a PP over $\gf(q^2)$ if and only if both $(a+b)x$ and $(a-b)x+(2x+u)^t\alpha^{t-1}$ are PPs over $\gf(q)$. It is easy to see that $(a-b)x+(2x+u)^t\alpha^{t-1}$ is a PP over $\gf(q)$ if and only if $(a-b)x+2x^t\alpha^{t-1}$ is a PP over $\gf(q)$. 
The desired conclusions in the second part then follow.  
\end{proof}

Let $q$ be a prime power with $q\equiv1\pmod{3}$. Then $3(q-1)|q^3-1$, so there is an element $\alpha\in\gf(q^3)\backslash\gf(q)$ and $\alpha^3\in\gf(q)$. Let $\alpha^3=b\in\gf(q)$ and $\omega=\alpha^{q-1}=b^{\frac{q-1}{3}}\in\gf(q)$. Then $\omega\ne1$ and $\omega^3=1, 1+\omega+\omega^2=0$, and $\{1, \alpha, \alpha^2\}$ is a basis of $\gf(q^3)$ over $\gf(q)$. Moreover, we have 
$$\alpha^q=\omega\alpha, \quad \alpha^{q^2}=\omega^2\alpha.$$

Similarly, we have the following. 

\begin{theorem} 
Let $q$ be a prime power with $q\equiv1\pmod{3}$ and let $t$ be a positive integer. Let $a, b, c, u\in\gf(q)$. Then 
$$F(x)=ax+bx^q+cx^{q^2}+(x+x^q+x^{q^2}+u)^t$$
is a PP over $\gf(q^3)$ if and only if $(a+b\omega+c\omega^2)(a+b\omega^2+c\omega)\ne0$ and $(a+b+c)x+3x^t$ is a PP over $\gf(q)$.
\end{theorem}

\begin{proof} 
Let $x=x_1+x_2\alpha+x_3\alpha^2$. Since $\alpha^q=\omega\alpha$ and $\alpha^{q^2}=\omega^2\alpha$, we have
$$F(x)=(a+b+c)x_1+(a+b\omega+c\omega^2)x_2\alpha+(a+b\omega^2+c\omega)x_3\alpha^2+(3x_1+u)^t.$$
By Theorem \ref{main-thm2.1}, $F(x)$ is a PP over $\gf(q^3)$ if and only if $(a+b\omega+c\omega^2)x$, $(a+b\omega^2+c\omega)x$ and $(a+b+c)x+(3x+u)^t$ are PPs over $\gf(q)$. It is easily seen that $(a+b+c)x+(3x+u)^t$ is a PP over $\gf(q)$ 
if and only if $(a+b+c)x+3x^t$ is a PP over $\gf(q)$. The desired conclusions then follow. 
\end{proof}

Let symbols and notations be the same as above. We have then the following. 

\begin{theorem} 
Let $q$ be a prime power with $q\equiv1\pmod{3}$ and let $t$ be a positive integer. Let $a, b, c, u\in\gf(q)$ and define 
$$
F_1(x)=ax+bx^q+cx^{q^2}+(x+\omega x^q+\omega^2x^{q^2}+u\alpha^2)^t.
$$ 
\begin{enumerate}
\item When $t\not\equiv1\pmod{3}$, $F_1(x)$ is a PP over $\gf(q^3)$ if and only if $(a+b+c)(a+b\omega+c\omega^2)(a+b\omega^2+c\omega)\ne0$. 
\item When $t\equiv1\pmod{3}$, $F_1(x)$
is a PP over $\gf(q^3)$ if and only if $(a+b+c)(a+b\omega^2+c\omega)\ne0$ and $(a+b\omega^2+c\omega)x+\alpha^{2(t-1)}3x^t$ is a PP over $\gf(q)$.
\end{enumerate} 
\end{theorem}

\begin{proof} 
Let $x=x_1+x_2\alpha+x_3\alpha^2$. Since $\alpha^q=\omega\alpha$ and $\alpha^{q^2}=\omega^2\alpha$, we have
$$F(x)=(a+b+c)x_1+(a+b\omega+c\omega^2)x_2\alpha+(a+b\omega^2+c\omega)x_3\alpha^2+(3x_3\alpha^2+u\alpha^2)^t.$$
Assume now that  $t\not\equiv1\pmod{3}$. We have then $\alpha^{2t}\ne d\alpha^2$ for any $d\in\gf(q)$. Let $x=x_1+x_2\alpha+x_3\alpha^2$ and $y=y_1+y_2\alpha+y_3\alpha^2$, $x_i, \, y_i \in \gf(q),\, i=1, 2, 3$ such that $F_1(x)=F_1(y)$. 

When $t \equiv 0 \pmod{3}$, we have 
\begin{eqnarray*}
\left\{ \begin{array}{l}
(a+b\omega+c\omega^2)x_2=(a+b\omega+c\omega^2)y_2, \\
(a+b\omega^2+c\omega)x_3=(a+b\omega^2+c\omega)y_3, \\
(a+b+c)x_1+(3x_3+u)^t\alpha^{2t}=(a+b+c)y_1+(3y_3+u)^t\alpha^{2t}.
\end{array}
\right. 
\end{eqnarray*} 
We then conclude that $x=y$ if and only if $(a+b+c)(a+b\omega+c\omega^2)(a+b\omega^2+c\omega)\ne0$; 
that is, $F_1(x)$ is a PP over $\gf(q^3)$ if and only if $(a+b+c)(a+b\omega+c\omega^2)(a+b\omega^2+c\omega) \ne 0$.

When $t\equiv2\pmod{3}$, we have 
\begin{eqnarray*}
\left\{ 
\begin{array}{l}
(a+b+c)x_1=(a+b+c)y_1,  \\
(a+b\omega^2+c\omega)x_3=(a+b\omega^2+c\omega)y_3, \\ 
(a+b\omega+c\omega^2)x_2+(3x_3+u)^t\alpha^{2t-1}=(a+b\omega+c\omega^2)y_2+(3y_3+u)^t\alpha^{2t-1}.
\end{array}
\right. 
\end{eqnarray*} 
We deduce that $x=y$ if and only if $(a+b+c)(a+b\omega+c\omega^2)(a+b\omega^2+c\omega)\ne0$; 
that is, $F_1(x)$ is a PP over $\gf(q^3)$ if and only if $(a+b+c)(a+b\omega+c\omega^2)(a+b\omega^2+c\omega) \ne 0$.

When $t \equiv 1\pmod{3}$, we have 
\begin{eqnarray*}
\left\{ \begin{array}{l}
(a+b+c)x_1=(a+b+c)y_1, \\ 
(a+b\omega+c\omega^2)x_2=(a+b\omega+c\omega^2)y_2, \\ 
(a+b\omega^2+c\omega)x_3+(3x_3+u)^t\alpha^{2t-2}=(a+b\omega^2+c\omega)y_3+(3y_3+u)^t\alpha^{2t-2}.
\end{array}
\right. 
\end{eqnarray*} 
By Theorem \ref{main-thm2.1},  $F(x)$ is a PP over $\gf(q^3)$ if and only if $(a+b\omega+c\omega^2)x$, $(a+b+c)x$ and $(a+b\omega^2+c\omega)x+(3x+u)^t\alpha^{2t-2}$ are PPs over $\gf(q)$; that is,  $(a+b+c)(a+b\omega^2+c\omega)\ne0$ and $(a+b\omega^2+c\omega)x+ \alpha^{2(t-1)}(3x+u)^t$ is a PP over $\gf(q)$. It is straightforward to see that 
$(a+b\omega^2+c\omega)x+ \alpha^{2(t-1)}(3x+u)^t$ is a PP over $\gf(q)$ if and only if 
$(a+b\omega^2+c\omega)x+ \alpha^{2(t-1)}3x^t$ is a PP over $\gf(q)$. The desired conclusion then follows. \end{proof}

\section{Concluding remarks} 

The contributions of this paper are the four recursive constructions of permutation polynomials over $\gf(q^2)$ with 
permutation polynomials over $\gf(q)$, and the construction of permutation polynomials over $\gf(2^{2m})$ with 
o-polynomials over $\gf(2^m)$. The five generic constructions give infinitely many new permutation polynomials. 

Although there are a number of references on o-polynomials, our coverage of o-polynomials in this paper contains some 
extensions of known families of o-polynomials. The reader is invited to settle the two conjectures on the extended families 
of o-polynomials.


\begin{thebibliography}{99}

\bibitem{AGW} A. Akbary, D. Ghioca, Q. Wang, On constructing permutations of finite fields, Finite Fields Appl. 17 (2011) 51--67.

\bibitem{CHZ}  X. Cao, L. Hu, Z. Zha, Constructing permutation polynomials from piecewise permutations, Finite Fields Appl. 26 (2014) 162--174.

\bibitem{CDY}
C. Carlet, C. Ding and J. Yuan, Linear codes from highly nonlinear functions and their secret sharing schemes,
IEEE Trans. Inform. Theory 51(6) (2005) 2089--2102.

\bibitem{Ch88} W. Cherowitzo, Hyperovals in Desarguesian planes of even order, Annals of Discrete Mathematics 37 (1988)  87--94.

\bibitem{Ch96} W. Cherowitzo, Hyperovals in Desarguesian planes: an update, Discrete Math. 155 (1996) 31--38.

\bibitem{Subiaco} W. Cherowitzo, T. Penttila, I. Pinneri,  G.F. Royle, Flocks and ovals,
Geometriae Dedicata 60 (1996) 17--37.

\bibitem{Adeil} W. Cherowitzo,  C. M. O'Keefe, and T. Penttila, A unified construction of finite geometries related to
q-clans in characteristic two, Adv. Geom. 3 (2003) 1--21. 

\bibitem{DY07}
C. Ding and J. Yin, Signal sets from functions with optimum nonlinearity, IEEE Trans. Communications 55(5)
(2007) 936--940.

\bibitem{DY06}
C. Ding and J. Yuan, A family of skew Hadamard difference sets, J. Comb. Theory Ser. A 113 (2006) 1526--1535.

\bibitem{Glynn83} D. G. Glynn, Two new sequences of ovals in finite Desarguesian planes of even order,
in: L.R.A. Casse (Ed.), Combinatorial Mathematics X, Lecture Notes in Mathematics 1036, Springer, 1983, pp. 217--229. 

\bibitem{Hirsh} J. W. P. Hirschfeld, Projective Geometries over Finite Fields, Oxford University Press, Oxford, 1998. 

\bibitem{Hou12} X. Hou, A new approach to permutation polynomials over finite fields, Finite Fields Appl. 18 (2012) 492--521. 

\bibitem{Houxd} X. Hou, Permutation polynomials over finite fields — A survey of recent advances, 
Finite Fields Appl. 32 (2015) 82--119. 

\bibitem{HMSY} X. Hou, G. L. Mullen, J. A. Sellers, J. L. Yucas, Reversed Dickson polynomials over finite fields, Finite
Fields Appl. 15 (2009) 748--773.

\bibitem{Laigle}
Y. Laigle-Chapuy, Permutation polynomials and applications to coding theory, Finite Fields Appl. 13 (2007) 58--70. 

\bibitem{Lidl85}
R. Lidl. On cryptosystems based on permutation polynomials and finite fields. In T. Beth, N. Cot, and
I. Ingemarsson, editors, Advances in Cryptology - EuroCrypt--84, volume 209 of Lecture Notes in Computer
Science, pages 10--15, Berlin, 1985. Springer-Verlag. 

\bibitem{LM84}
R. Lidl and W. B. Muller, Permutation polynomials in RSA-cryptosystems, in: Advances in Cryptology, Plenum,
New York, 1984, 293--301. 

\bibitem{LMT93}
R. Lidl, G. L. Mullen and G. Turnwald, Dickson Polynomials, Longman Scientific  Technical, Harlow,
Essex, UK, 1993.

\bibitem{LN97} R. Lidl and H. Niederreiter, Finite Fields, Cambridge University Press, Cambridge, 1997. 

\bibitem{Mullen89}
G. L. Mullen. Permutation polynomials and nonsingular feedback shift registers over finite fields. IEEE
Trans. Information Technology, 35(4) (1989) 900--902.

\bibitem{Okeefe} C. M. O'Keefe and T. Penttila, Polynomials for hyperovals of Desarguesian plane, 
J. Austral. Math. Soc. (Series A) 51 (1991) 436--447.  

\bibitem{Pay85} S. E. Payne, A new infinite family of generalized quadrangles, Congressus Numerantium 49 (1985) 115--128. 

\bibitem{SH98}
J. Schwenk and K. Huber, Public key encryption and digital signatures based on permutation polynomials, Electronic
Letters 34 (1998) 759--760.

\bibitem{RSA}
R. L. Rivest, A. Shamir, and L. M. Adelman, A method for obtaining digital signatures and public-key cryptosystems,
Comm. ACM 21 (1978) 120--126.

\bibitem{Segre57} B. Segre, Sui k-archi nei piani finiti di caratteristica 2, Revue de Math. Pures Appl. 2 (1957) 289--300.

\bibitem{Segre62} B. Segre, Ovali e curvenei piani di Galois di caratteristica due, Atti Accad. Naz. Lincei Rend. (8) 32 (1962)  785--790.

\bibitem{SegreBartocci} B. Segre, U. Bartocci, Ovali ed alte curve nei piani di Galois di caratteristica due, Acta Arith., 18 (1971)  423--449. 


\bibitem{ST05}
J. Sun and O. Y. Takeshita. Interleavers for turbo codes using permutation polynomials over integer
rings. IEEE Trans. Information Theory, 51(1) (2005) 101--119. 

\bibitem{TZH} Z. Tu, X. Zeng, L. Hu, Several classes of complete permutation polynomials, Finite Fields Appl. 25
(2014) 182--193.

\bibitem{WY12} Q. Wang, J.L. Yucas, Dickson polynomials over finite fields, Finite Fields Appl. 18 (2012) 814–831.

\bibitem{YCD06}
J. Yuan, C. Carlet and C. Ding, The weight distribution of a class of linear codes from perfect nonlinear functions,
IEEE Trans. Inform. Theory 52(2) (2006) 712--717. 

\bibitem{ZZH} X. Zeng, X. Zhu, L. Hu, Two new permutation polynomials with the form $(x^{2^k} +x+\delta)^s +x$ 
over $F_{2^n}$, Appl. Algebra Eng. Commun. Comput. 21 (2010) 145--150.

\bibitem{ZhaHuCao} Z. Zha, L. Hu and X. Cao, Constructing permutations and complete permutations over finite fields
via subfield-valued polynomials, Finite Fields and Their Applications, Vol. 31, pp. 162--177,  January 2015.

\end{thebibliography}
\end{document}